\documentclass[11pt]{article}

\usepackage[utf8]{inputenc}
\usepackage[T1]{fontenc}
\usepackage{lmodern}
\usepackage[left=1in,right=1in,top=1in,bottom=1.25in]{geometry}

\usepackage{microtype}

\usepackage{amsmath}
\usepackage{amssymb}
\usepackage{amsthm}
\usepackage{thmtools}
\usepackage{thm-restate}

\usepackage{titlesec}
\usepackage{fancybox}
\usepackage{xcolor}
\usepackage{xspace}
\usepackage{enumerate}

\usepackage[ocgcolorlinks]{hyperref} 
\usepackage{cleveref}
\usepackage{authblk}
\usepackage{algorithm}

\titlespacing{\paragraph}{%
  0pt}{
  0.5\baselineskip}{
  1em}

\makeatletter

\def\moverlay{\mathpalette\mov@rlay}
\def\mov@rlay#1#2{\leavevmode\vtop{%
		\baselineskip\z@skip \lineskiplimit-\maxdimen
		\ialign{\hfil$\m@th#1##$\hfil\cr#2\crcr}}}
\newcommand{\charfusion}[3][\mathord]{
	#1{\ifx#1\mathop\vphantom{#2}\fi
		\mathpalette\mov@rlay{#2\cr#3}
	}
	\ifx#1\mathop\expandafter\displaylimits\fi}
\makeatother

\colorlet{DarkRed}{red!75!black}
\colorlet{DarkGreen}{green!75!black}
\colorlet{DarkBlue}{blue!75!black}

\hypersetup{
	linkcolor = DarkRed,
	citecolor = DarkGreen,
	urlcolor = DarkBlue,
	bookmarksnumbered = true,
	linktocpage = true
}

\declaretheorem[numberwithin=section]{theorem}
\declaretheorem[numberlike=theorem]{lemma}

\declaretheorem[numberlike=theorem]{corollary}
\declaretheorem[numberlike=theorem]{definition}

\declaretheorem[numberlike=theorem]{fact}

\newenvironment{fminipage}%
  {\begin{Sbox}\begin{minipage}}%
  {\end{minipage}\end{Sbox}\fbox{\TheSbox}}

\def\defeq{\stackrel{\mathrm{def}}{=}}

\def\abs#1{\left|#1  \right|}

\newcommand\cchi{\boldsymbol{\chi}}

\newcommand\ttau{\boldsymbol{\tau}}

\newcommand\ttautilde{\widetilde{\boldsymbol{\tau}}}

\newcommand\dd{\boldsymbol{\mathit{d}}}

\newcommand\rr{\boldsymbol{\mathit{r}}}

\newcommand\rrApprox{\boldsymbol{\widetilde{\mathit{r}}}}

\newcommand\ww{\boldsymbol{\mathit{w}}}
\newcommand\yy{\boldsymbol{\mathit{y}}}
\newcommand\zz{\boldsymbol{\mathit{z}}}
\newcommand\xx{\boldsymbol{\mathit{x}}}

\renewcommand\AA{\boldsymbol{\mathit{A}}}
\newcommand\BB{\boldsymbol{\mathit{B}}}

\newcommand\DD{\boldsymbol{\mathit{D}}}

\newcommand\GG{\boldsymbol{\mathit{G}}}
\newcommand\HH{\boldsymbol{\mathit{H}}}

\newcommand\MM{\boldsymbol{\mathit{M}}}
\newcommand\LL{\boldsymbol{\mathit{L}}}

\newcommand\Otil{\widetilde{O}}
\newcommand\Ohat{\widehat{O}}

\newcommand\eps{\varepsilon}

\newcommand{\er}{\mathcal{R}_{\mathrm{eff}}}

\newcommand*\samethanks[1][\value{footnote}]{\footnotemark[#1]}

\begin{document}

\title{Density Independent Algorithms
for Sparsifying $k$-Step Random Walks}

\author[1]{Gorav Jindal}
\author[1]{Pavel Kolev\thanks{This work has been funded by the Cluster of Excellence ``Multimodal Computing and Interaction'' within the Excellence Initiative of the German Federal Government.}}
\author[2]{Richard Peng\thanks{This material is based upon work supported by the National Science Foundation under Grant No. 1637566.}}
\author[2]{Saurabh Sawlani\samethanks}

\affil[1]{Max-Planck-Institut f\"{u}r Informatik, Saarbr\"{u}cken, \{gjindal,pkolev\}@mpi-inf.mpg.de}
\affil[2]{Georgia Institute of Technology, \{rpeng,sawlani\}@gatech.edu}

\date{}
\maketitle

\begin{abstract}
We give faster algorithms for producing sparse approximations
of the transition matrices of $k$-step random walks
on undirected, weighted graphs.
These transition matrices also form graphs, and arise
as intermediate objects in a variety of graph algorithms.
Our improvements are based on a better understanding of processes
that sample such walks, as well as tighter bounds on key
weights underlying these sampling processes.
On a graph with $n$ vertices and $m$ edges, our algorithm
produces a graph with about $n\log{n}$ edges that approximates
the $k$-step random walk graph in about $m + n \log^4{n}$ time.
In order to obtain this runtime bound, we also revisit ``density independent''
algorithms for sparsifying graphs whose runtime overhead is expressed only
in terms of the number of vertices.
\end{abstract}


\setcounter{page}{0}
\pagenumbering{arabic}

\section{Introduction}
\label{sec:introduction}	

Random walks in graphs are fundamental objects in both graph
algorithms and graph data structures.
Problems related to random walks,
such as shortest paths and minimum cuts
are well studied in both static~\cite{Sommer14}
and dynamic settings~\cite{HenzingerKN14,GoranciHT16}.
While some of these problems, such as shortest path, aim to find a single walk,
other problems such as flows/cuts~\cite{GoldbergT14}
or triangle densities~\cite{BjorklundPWZ14,Tsourakakis08}
aim to capture information related to collections of walks.
Algorithms and data structures for such problems often need
to store, or can be sped up by, intermediate structures that
capture the global properties of multi-step walks~\cite{PaghT12,GuptaP13,AbrahamDKKP16,BhattacharyaHI14}.
However, many intermediate structures are inherently dense
and therefore expensive to compute explicitly.

Graph sparsification is a technique for efficiently
approximating a dense graph by a sparser one,
while preserving some key properties such as sizes of
graph cuts, distances between vertices,
or linear operator properties of matrices associated with the graphs.
Spectral sparsifiers provide linear operator approximations that 
also imply approximation to all graph cuts. 
Their constructions have some of the simplest interactions with
statistical concentration results~\cite{BatsonSST13}.
Spectral sparsifiers of (possibly dense) intermediate objects
have a variety of applications in graph algorithms,
such as sampling from graphical models~\cite{ChengCLPT15},
solving linear systems~\cite{PengS14,KyngLPSS16},
sampling random spanning trees~\cite{DurfeeKPRS16}
and maintaining approximate minimum cuts in dynamically
changing graphs~\cite{AbrahamDKKP16}.
In these applications, the optimal performance is achieved
by producing a sparsifier of the denser intermediate object
directly, instead of generating the larger exact object.
Such implicit sparsification routines were first studied for
random walk matrices~\cite{PengS14,ChengCLPT15}.
These matrices contain the pairwise transition probabilities
between vertices under $k$-step walks. Moreover, such matrices are dense even
for sparse original graphs with small $k$:
the $2$-step walk on the $n$-vertex star has non-zero
transition probabilities between any pair of vertices.
On the other hand, as the $k$-step random walk can be viewed
as a single random process, these vertex-to-vertex transition
probabilities correspond to a graph, and therefore have a
sparse approximate.

Cheng et al.~\cite{ChengCLPT15} systematically studied random
walk sparsification and its applications.
They gave a routine that produces an $\epsilon$-spectral
sparsifier (which we will formally define in Subsection~\ref{subsec:SApproxG})
with $O(\epsilon^{-2} n \log{n})$ edges for a $k$-step walk
matrix in $O( \epsilon^{-2} k^2 m \log^{O(1)}n)$ time.
Our main result, which we show in Section~\ref{sec:WalkSparsify}
is a direct improvement of that routine:
\begin{theorem}\label{thm:main}
(\textbf{Sparsifying Laplacian Monomials})
Given a graph $G$
and an error $\eps\in(0,1)$, there is an algorithm that outputs
an $\eps$-spectral sparsifier of $G^{k}$ with
at most $O(\eps^{-2}n\log n)$ edges in
$\Ohat(m+k^{2}\eps^{-2}n\log^{4}n)$ time.
\footnote{
We use $\Ohat$ to denote the omission of
logarithmic terms lower than the ones shown in the set.
In all cases in this paper, we track terms of $\log{n}$
explicitly and such notation hides terms of $\log\log{n}$.
In all these cases, this notation hides a term of at most
$(\log\log{n})^2$.}
\end{theorem}
We term this type of running time with most of the overhead
on the number of vertices, $n$, as density independent.
Such runtimes arise naturally in many other graph problems~\cite{FredmanT87},
and was first studied for graph sparsification in an earlier manuscript by
a subset of the authors~\cite{JindalK15:arxiv}.
Our results can also be combined with the repeated-squaring technique
in~\cite{ChengCLPT15} to reduce the runtime dependence on $k$
to logarithmic~\cite{ChengC16:comm}.
This plus generalizations to general random walk polynomials~\cite{ChengC16:comm}
would then supersede all claims from~\cite{JindalK15:arxiv}.
As these steps are much closer to~\cite{ChengCLPT15},
we will focus on the small $k$ case in this paper.
Furthermore, as our sparsification algorithm has a much more
direct interaction with routines that provide upper bounds of
effective resistances, they can likely be combined with tools
from~\cite{AbrahamDKKP16} to give dynamic algorithms for
maintaining $G^{k}$ under insertions/deletions to $G$.
However, as there are currently only few applications
of such sparsifiers, we believe it may be more fruitful to
extend the applications before further developing the tools.

Our algorithms, as with the ones from~\cite{JindalK15:arxiv,ChengCLPT15}
are based on implicit sampling of dense graphs by probabilities
related to effective resistances.
Such approach is the only known efficient approach for even the
`simpler' problem of producing cut sparsifiers of $G^{k}$.
Our improvements rely on an a key insight from the sparse Gaussian elimination
algorithm by Kyng and Sachdeva~\cite{KyngS16}: using triangle inequality
between effective resistances to obtain a tighter set of probability upper bounds.
This allows us to select the first edge, from which we ``grow'' a length $k$ walk, via an adaptive sampling process, instead of using uniform sampling as in the previous result~\cite{ChengCLPT15}.
Furthermore, this adaptive process removes any sampling count dependencies on $m$,
the number of edges, making a density-independent runtime possible.
This type of running time also has analogs in input sparsity
time algorithms in randomized numerical linear algebra~\cite{Woodruff14,CLMMPS15,CMM17}.
Obtaining density-independent bounds is critical for graph sparsification
algorithms because they are primarily invoked
on relatively dense graphs.
A graph sparsification routine that produces a sparsifier
with $\Ohat(n \log^2{n})$ edges in $\Ohat(m \log^2{n})$ time,
such as the combinatorial algorithm given in~\cite{KyngPPS17},
will only be invoked when $m > n \log^2{n}$, which means
the running time of the algorithm is actually
$\Omega(n \log^4{n})$.
As a result, we believe that for graph sparsification to work
as a primitive for processing large graphs, a running time of
$\Ohat(m + n\log^2{n})$ or better is necessary.

In Section~\ref{sec:mPlusStuff}, we provide some steps toward
this direction by giving a better density-independent spectral
sparsification algorithm.
We combine ideas from previous density-independent algorithms for
sparsifying graphs~\cite{KoutisLP15} with recent developments
in tree embedding and numerical algorithms to
obtain numerical sparsification routines
that run in $\Ohat(m + n\log^4{n})$ time, and combinatorial
ones that take $\Ohat(m + n\log^6{n})$ time.
Both of these routines are in turn applicable to the walk
sparsification algorithm in Section~\ref{sec:WalkSparsify},
giving routines for sparsifying $k$-step walks with similar running times:
the bound stated in Theorem~\ref{thm:main} is via the numerical routine.
While these results are far from what we think is the best possible,
we show a variety of new algorithmic tools for designing algorithms
for sparsifying $k$-step random walks.

\section{Background}
\label{sec:background}

We start with some background information about
graphs and matrices corresponding to them.
These matrices allow us to define graph approximations,
as well as compute key sampling probabilities needed
to produce spectral sparsifiers.
Due to space constraints, we will only formally define
most of the concepts.
More intuition on them can be found in notes on spectral
graph theory and random walks
such as~\cite{DoyleS84:book,Lovasz93}.

\subsection{Random Walks and Matrices}

Let $G=(V,E,\ww)$ be a weighted undirected graph.
We define its adjacency matrix $\AA$ as
$\AA_{uv} \defeq \ww_{uv}$, and its degree matrix
$\DD$ as $\DD_{uu} \defeq \sum_{v \in V} w_{uv}$
and $\DD_{uv} \defeq 0$ when $u \neq v$.
This leads to the graph Laplacian $\LL_G \defeq \DD - \AA$.

One step of a random walk can be viewed as distributing
the `probability mass' at a vertex evenly among the edges
leaving it, and passing them onto its neighbors.
In terms of these matrices, it is equivalent to first
dividing by $\DD$, and then multiplying by $\AA$.
Thus, the transition matrix of the $k$\textsuperscript{th}
step random walk is given by $(\DD^{-1}\AA)^{k}$.
The corresponding Laplacian matrix of the $k$-step random walk is defined by
\[
\LL_{G^{k}} \defeq \DD - \AA \left( \DD^{-1} \AA \right)^{k-1}.
\]
The matrices $\AA( \DD^{-1} \AA)^{k-1}$ can be viewed as a sum over length
$k$ walks.
This view is particularly useful in our algorithm,
as well as the earlier walk sparsification algorithm
by Cheng et al.~\cite{ChengCLPT15} because these walks
are a more `natural' unit upon which sparsification by
effective resistances is applied.
Formally, we can define the weight of a length $k$ walk
$u_{0},\dots,u_{k}$ by
\begin{equation}\label{eq:weightGk}
\ww_{u_{0},\dots,u_{k}}
\defeq
\frac{\prod_{i=1}^{k}\ww_{u_{i-1}, u_{i}}}
	{\prod_{i=1}^{k-1}\dd_{u_{i}}}.
\end{equation}
Straightforward checking shows that for any $u_0,u_k\in V$, it holds that
\begin{equation}\label{eq:weightGkEdge}
\ww^{\GG^{k}}_{u_0,u_k}
\defeq \left[ \AA\left(\DD^{-1}\AA\right)^{k-1} \right]_{u_0u_k}
=\sum_{ u_{1},\dots, u_{k-1}}
  \ww_{u_{0},\dots, u_{k}}.
\end{equation}

\subsection{Spectral Approximations of Graphs}\label{subsec:SApproxG}

Our notion of matrix approximations will be through the
$\approx$ symbol, which is in turn defined through
the L\"{o}ewner partial ordering of matrices.
For two matrices, $\AA$, and $\BB$, we say that
\[
\AA \preceq \BB
\]
if $\BB - \AA$ is positive semidefinite, and
\[
\AA \approx_{\kappa} \BB
\]
if there exists bounds
$\lambda_{\min}$ and $\lambda_{\max}$ such that
$\lambda_{\min} \AA \preceq \BB \preceq \lambda_{\max} \AA$,
and $\lambda_{\max} \leq \kappa \lambda_{\min}$.
This notation is identical to generalized eigenvalues,
and in particular, $\LL_{G} \approx_{\kappa} \LL_H$
implies that all cuts on them are within a factor
of $\kappa$ of each other.

The adjacency matrix of a graph has several of undesirable
properties when it comes to operator based approximations:
it can have a large number of eigenvalues at $0$, which
must be exactly preserved under relative error approximations.
As a result, graph approximations are defined in terms of
graph Laplacians.
As we will discuss below, these approximations are often
in terms of reducing edges.
So formally, we say that a graph $H$ is a $\kappa$-sparsifier
of $G$ if
\[
\LL_H \approx_{\kappa} \LL_G,
\]
and our goal is to compute an $\epsilon$-sparsifier of the $k$-step random walk matrix
$\LL_{G^{k}}$.

\subsection{Graph Sparsification by Effective Resistances}\label{subsec:GSbyEffRes}

There are two ways of viewing graph sparsification:
either as tossing coins independently on the edges,
or sampling a number of them from an overall probability
distribution.
We take the second view here because it may prove expensive to
access all edges in $G^{k}$.
The pseudocode of the generic sampling scheme is given in
Algorithm~\ref{alg:IdealSample}.

\begin{algorithm}[h]
	\caption{$\mathrm{IdealSample}(G, \eps, \ttautilde)$}
	\label{alg:IdealSample}
	
	\textbf{Input:} Graph $G=(V, E, \ww)$, integer $k$,
	and resistance upper bounds $\ttautilde_e$ such that $\ttautilde_e \geq \ww_e \er^{G}(e) $
	for all edges $e$.
	
	\textbf{Output: } An $\eps$-sparsifier $\HH$ of $\LL_G$ with $O(\eps^{-2} \sum_{e\in E} \ttautilde_e \log{n})$ edges.
	
	\begin{enumerate}
		\item Initiate $\HH$ as an empty graph.
		\item Set sample count $N \leftarrow O(\eps^{-2} \sum_{e\in E} \ttautilde_e \log{n})$
		\item Repeat $N$ times:
		\begin{enumerate}
			\item Pick an edge $e$ in $G$ with probability proportional to $\ttautilde_e$.
			\item Add $e$ to $\HH$ with new weight $\ww^{G}_e / (\ttautilde_e N)$.
		\end{enumerate}
	\end{enumerate}
\end{algorithm}

Algorithmically, the sampling step can be implemented
by first generating a number uniformly random in
$[0, \sum_{e} \ttautilde_e]$, and binary searching among
the prefix sums of the $\ttautilde_e$ values until it reaches
the edge corresponding to that point.
Note also that if we want to generate random numbers
with bounded precision, we can also round the $\ttautilde_e$
values of all edges up to the nearest multiple of $1/n$,
leading to at most $m / n = O(n)$ extra edges.
The guarantees of this routine require defining effective
resistances and leverage scores.
Effective resistance is a metric on a graph that is defined by:
\begin{equation}\label{eq:defER}
\er^{G}(u, v) \defeq \cchi_{uv}^{T} \LL_G^{\dag} \cchi_{uv},
\end{equation}
where $\LL_G^{\dag}$ denotes the pseudoinverse of $\LL_G$ and $\cchi_{uv}$ is the indicator vector with $1$ at $u$
and $-1$ at $v$.

The effective resistances $\er^{G}$ are directly related to the statistical leverage scores
$\ttau$ by the relation $\ttau_e=\ww_e \er^{G}(e)$. Moreover, these scores are well defined
for general matrices, and have a wide range of applications in randomized linear
algebra~\cite{Woodruff14,CLMMPS15,CMM17}.
The guarantees of sampling by weight times effective resistance,
or leverage scores, can then be formalized as:
\begin{lemma}\label{lem:sparsify}
(Sampling by Upper Bounds on Leverage Scores~\cite{Tropp12})
Suppose $G=(V,E,w)$ is a graph and $\ttautilde$ is a vector such that
$\ttautilde_e \geq \ww_e \er^{G}(e)$ for every edge $e$,
then any process that simulates the ideal sampling in Algorithm~\ref{alg:IdealSample}
produces an $\eps$-sparsifier of $G$ with $O(\eps^{-2} \sum_{e} \ttautilde_e \log{n})$ edges
in $\Otil(\eps^{-2} \sum_{e} \ttautilde_e \log^2{n})$ time.
\end{lemma}

The bound on sample count then follows from:
\begin{fact}[Foster's Theorem]\label{fact:foster}
For any undirected graph $G = (V, E, \ww)$, we have:
\[
\sum_{e\in E} \ww_e \er^{G}\left(u, v\right) = n - 1.
\]
\end{fact}
Leverage scores are the preferred objects for defining
sampling distributions as they are scale invariant:
doubling the weights of all edges does not change them.
However, we will still make extensive uses of effective
resistances because of the need to approximate them
across different graphs.
Such approximations are difficult to state for leverage scores
because spectrally similar graphs may have very different
sets of combinatorial edges.
\begin{fact}
	\label{fact:ERApprox}
	If $G$ and $H$ are graphs such that $\LL_G \preceq \LL_H$,
	then for any vertices $u$ and $v$ we have
	\[
		\er^{H}(u, v) \leq \er^{G}(u, v).
	\]
\end{fact}

Note that this generalizes Rayleigh's monotonicity law, which postulates
that the effective resistances can only increase as one removes edges from a graph.

\section{Random Walk Sparsification via Walk Sampling}
\label{sec:WalkSparsify}

In this section we describe our improved algorithm for
sparsifying random walk polynomials.
The main difficulty faced by such a routine is that the
actual walk matrix cannot be constructed.
Instead, we need to simulate the ideal sampling routine
shown in Algorithm~\ref{alg:IdealSample} by constructing
nearly tight upper bounds of leverages scores in $G^{k}$
that can also be efficiently sampled from,
without having explicit access to $G^{k}$.

The first obstacle to obtain such estimates is to get an access
to effective resistances in $G^{k}$.
To this end, the following lemma from~\cite{ChengCLPT15}
provides a helpful starting point.
\begin{lemma}\cite{ChengCLPT15}\label{lem:GndGk}
For odd $k$, we have
$\frac{1}{2} \LL_G \preceq \LL_{G^k} \preceq k \LL_G$
and for even $k$, we have
$\LL_{G^2} \preceq \LL_{G^k} \preceq \frac{k}{2} \LL_{G^2}.$
\end{lemma}
Furthermore, note that Lemma~\ref{lem:GndGk} combined with Fact~\ref{fact:ERApprox} implies for odd $k$ that
\begin{equation}\label{eq:erGkerG}
\er^{G^{k}}(u, v) \leq 2 \er^{G}(u, v)
\end{equation}
and for even $k$ that
\begin{equation}\label{eq:erGkerG2}
\er^{G^{k}}(u, v) \leq \er^{G^2}(u, v).
\end{equation}

Since $G^{k}$ might be dense, i.e. $E[G^{k}]=\Theta(n^2)$, it is prohibitive to use (\ref{eq:erGkerG}) and (\ref{eq:erGkerG2}) directly. Instead, we upper bound the values with a walk using the triangle inequality of effective resistances~\cite[Lemma 9.6.1]{S:Lecture07}.
\begin{fact}[Triangle Inequality for Effective Resistances]\label{fact:TriangleIneq}
For any graph $G$ and any walk $(u_0, u_1, \ldots, u_k)$, we have
\begin{equation}\label{eq:ErG}
\er^{G}(u_0, u_k) \leq
	\sum_{0 \leq i < k} \er^{G}(u_i, u_{i + 1}).
\end{equation}
\end{fact}
Now, suppose we have a vector $\rrApprox$ that upper bounds the effective resistances, i.e., $\rrApprox_e\geq\er^{G}(e)$ for all $e$. Then, by Lemma~\ref{lem:sparsify} and Fact~\ref{fact:TriangleIneq}, to sparsify $G^k$, it suffices to sample a length $k$ random walk in $G$ with probability proportional to
\begin{equation}\label{eq:weightErGkApprox}
\ttautilde_{u_0, u_1, \ldots, u_k }^{(k)}
\defeq
\ww_{ u_0, u_1, \ldots, u_k } \cdot \sum_{0 \leq i < k} \rrApprox_{u_i, u_{i + 1}}.
\end{equation}
This distribution has the advantage that it is efficiently computable:
\begin{lemma}
\label{lem:DistributionAccess}
For any graph $G = (V, E, \ww)$,
and any vector $\rrApprox\in\mathbb{R}^E$,
we can sample length $k$ walks with probability proportional to
\[
\ww_{u_0, \ldots, u_k} \cdot \sum_{0 \leq i < k} \rrApprox_{u_i, u_{i + 1}}
\]
using the following procedure:
\begin{enumerate}
	\item Pick uniformly at random an index $i$ in the range $[0, k - 1]$.
	\item Choose an edge $(u_{i}, u_{i + 1})$ with probability proportional to $\ww_e\rrApprox_e$.
	\item Extend the walk in both directions from $u_{i}$
	and $u_{i + 1}$ via two random walks.
\end{enumerate}
\end{lemma}

\begin{proof}
    By total law of probability, the procedure samples a fixed walk $(u_0,\ldots,u_k)$ with probability equal to
    \begin{equation*}
        \sum_{i=0}^{k-1}\frac{1}{k} \cdot \frac{\ww_{u_{i},u_{i+1}}\rrApprox_{u_{i},u_{i+1}}}{\langle\ww,\rrApprox\rangle} \cdot \prod_{j=1}^{i} \frac{\ww_{u_{j-1},u_{j}}}{\dd_{u_{j}}}\cdot \prod_{j=i+1}^{k-1}\frac{\ww_{u_{j},u_{j+1}}}{\dd_{u_{j}}}= \frac{\ww_{u_{0},\ldots,u_{k}}}{k\langle\ww,\rrApprox\rangle} \sum_{i=0}^{k-1}\rrApprox_{u_{i},u_{i+1}},
    \end{equation*}
    where the first term is step (1), the second term is step (2) and the third and the fourth are for the two random walks extending the selected edge.
\end{proof}

The total number of samples needed by Lemma~\ref{lem:sparsify} can be extracted from summing over random walks containing
a particular edge in a way similar to~\cite{ChengCLPT15}.
For completeness, we present its proof in Appendix~\ref{sec:deferred}.
\begin{restatable}{lemma}{TotalWeights}
\label{lem:totalWeights}
	For any weighted graph $G = (V, E, \ww)$, any $k\in\mathbb{N}_+$,
	and any vector $\rrApprox\in\mathbb{R}^E$, it holds
	\begin{equation}\label{eq:UBnumSamples}
	\sum_{\left(u_0, u_1, \ldots, u_k \right)}
	\left[
	\ww_{u_0, u_1, \ldots, u_k}
	\cdot \left( \sum_{0 \leq i < k} \rrApprox_{u_i, u_{i + 1}} \right) \right]
	= k \cdot \left( \sum_{e\in E[G]}\ww_e \rrApprox_e \right).
	\end{equation}
\end{restatable}

For every odd $k$, by setting $\rrApprox$ to (an approximation of) $\er^{G}$, yields an efficient sampling procedure
due to (\ref{eq:UBnumSamples}) and Lemma~\ref{lem:DistributionAccess}.

However, when $k$ is even Lemma~\ref{lem:GndGk} gives a bound in terms of $\er^{G^2}$ (not $\er^{G}$), i.e. $\er^{G^{k}}(u, v)\leq \er^{G^2}(u, v)$. Hence, the distribution in Lemma~\ref{lem:DistributionAccess} requires an access to the $2$-step random walk matrix $G^{2}$, which might also be dense and therefore expensive
to compute.

Moreover, suppose $G$ is a $2$-length path graph $u-v-w$, then $\er^{G^{2}}(e)=+\infty$ for $e\in G$, since $G^{2}$ has only one edge $(u,w)$ (and self-loops).
A naive approach to tackle these issues is to substitute $\er^{G^{2}}$ with $\er^{G}$.
However, this approach fails shortly since it is not true in general that
\begin{equation}\label{eq:notTrue}
\er^{G}(u,v)+\er^{G}(v,w)\geq\er^{G^{2}}(u,w).
\end{equation}

We work around this by using effective resistances from the ``double cover'' of $G$, instead.
The ``double cover'' $G\times P_2$ is the tensor product of $G$ and a path of length $2$.
Combinatorially, $G \times P_2$ is a bipartite graph with vertex sets $V^{(A)}, V^{(B)}$
each a copy of $V$ such that for every edge $(u,v)\in G$ we insert in $G \times P_2$ the following two edges: $u^{(A)}v^{(B)}$ and $u^{(B)}v^{(A)}$ with $\ww_{u^{(A)}v^{(B)}} = \ww_{u^{(B)}v^{(A)}} = \ww_{uv}$.

The next lemma (proved in Appendix~\ref{sec:deferred}) fixes (\ref{eq:notTrue}) and guarantees for every edge $(u,w)\in G^2$ that
\begin{equation}\label{eq:gotItRight}
\er^{G^{2}}(u,w) = \er^{G\times P_2}(u^{(A)},w^{(A)})\leq\er^{G\times P_2}(u^{(A)},v^{(B)})+\er^{G\times P_2}(v^{(B)},w^{(A)}).
\end{equation}

\begin{restatable}{lemma}{EREven}
\label{lem:EREven}
For any vertices $u$ and $v$ in $G$, it holds
\[
\er^{G^2}(u, v)
= \er^{G \times P_2}(u^{\left(A\right)}, v^{\left(A \right)}),
\]
where $u^{(A)}$ and $v^{(A)}$ are the corresponding copies
of $u$ and $v$ in $V^{(A)}$, respectively.
\end{restatable}

Using the preceding results, we design an algorithm with improved sampling count.
It takes any procedure that produces effective resistance distribution that dominates
the true one, and produces samples that suffice for simulating the ideal sampling algorithm on $G^{k}$ (c.f. Subsection~\ref{subsec:GSbyEffRes}, Algorithm~\ref{alg:IdealSample}). The pseudocode for this routine is shown in Algorithm~\ref{alg_RW}.
\begin{algorithm}[h]
	\caption{$\mathrm{SparsifyG^{k}} \left(G, k, \eps,
		\textsc{EREstimator} \right)$}
	\label{alg_RW}
	
	\textbf{Input:} Graph $G=(V,E,w)$, integer $k$,
	error $\eps$, routine \textsc{EREstimator} that estimates
	upper bounds for effective resistances of a graph $G$.
	
	\textbf{Output: } An $\eps$-sparsifier of $G^{k}$
	
	\begin{enumerate}
		\item If $k$ is odd
			\begin{enumerate}
				\item set $\rrApprox \leftarrow \textsc{EREstimator}(G)$,
			\end{enumerate}
		\item else
			\begin{enumerate}
				\item Set $\rrApprox^{(2)} \leftarrow \textsc{EREstimator}(G \times P_2)$,
				\item Set $\rrApprox_{e} \leftarrow \rrApprox^{(2)}(u^{(A)}, v^{(B)})$,
				for every edge $e = uv\in E[G]$\quad\quad\quad(c.f. Lemma~\ref{lem:EREven}).
			\end{enumerate}
		
		\item Set sampling overhead $h \leftarrow O(\eps^{-2}\log{n})$
		and $N \leftarrow h \cdot k \sum_{e\in E[G]} \ww_e\rrApprox_e$.
		\item Repeat $N$ times
		
		\begin{enumerate}
			
			\item Pick an edge $e$ in $G$ with probability
			proportional to $\ww_e\rrApprox_e$.
					
			\item Pick a random integer $0 \leq i < k$ uniformly random,
			set $u_i$ and $u_{i + 1}$ to be the two endpoints of $e$.
			
			\item Complete this random walk by taking $k - 1 - i$
			steps of random walk from $u_{i + 1}$ and $i$ steps from $u_{i}$.
			
			\item Add the edge $(u_{0}, u_1 \ldots, u_k)$ to $H$ with weight (c.f. Eq.(\ref{eq:weightGk}))
            \[
                \frac{1}{h \cdot \ww_{u_{0},u_{1},\ldots,u_{k}}\cdot\sum_{0\leq i<k}\rrApprox_{u_{i}u_{i+1}}}.
            \]
		\end{enumerate}			
	\end{enumerate}
\end{algorithm}

Note that from the perspective of this framework of picking edges
with probabilities proportional to $\ww_e\rrApprox_e$, and extending
them into walks, the previous result~\cite{ChengCLPT15} can be viewed as utilizing
a simple $\mathrm{EREstimator}$ that returns $1$ on the effective resistance of every edge.

\begin{theorem}
\label{thm:SparsifyGk}
Given any graph $G$, any values of $k$ and $\eps$, and
any effective resistance estimation algorithm $\mathrm{EREstimator}$
that produces w.h.p. effective resistance that sum up to $f(n, m)$,
calling $\mathrm{SparsifyG^{k}}(G, k, \eps, \mathrm{EREstimator})$ produces an
$\eps$-sparsifier of $G^{k}$ with $O(\eps^{-2} k \log{n}\cdot f(2n, 2m) )$
edges in time proportion to the cost of one call to $\mathrm{EREstimator}$
on a graph of twice the size, plus an overhead of
$\Otil(\eps^{-2} k^2 \log^2{n}\cdot f(2n, 2m) )$.
\end{theorem}

\begin{proof}
By Lemma~\ref{lem:sparsify}, it suffices to show that this
algorithm simulates the ideal sampling algorithm given
in Algorithm~\ref{alg:IdealSample}.
Once again we split into the cases of $k$ being odd or even.

In the case of $k$ being odd, Lemma~\ref{lem:DistributionAccess}
gives that a walk $(u_0, u_1, \ldots u_k)$ is sampled with weight
at least
\[
\ww_{u_0, u_1, \ldots u_k}
\sum_{0 \leq i < k} \er^{G}\left( u_i, u_{i + 1}\right).
\]
The quality of the distribution
produced, follows from Lemma~\ref{lem:GndGk} and the triangle
inequality in Fact~\ref{fact:TriangleIneq}, which then combined
with Lemma~\ref{lem:sparsify} gives the quality of the output.
Also, the total size of the sparsifier, as well as the running time
follows from Lemma~\ref{lem:totalWeights}.

When $k$ is even, by combining Lemmas~\ref{lem:GndGk} and Lemma~\ref{lem:EREven}
we have
\[
\er^{G^{k}}\left(u, v\right)
\leq \er^{G \times P_2}
\left(u^{\left(A\right)}, v^{\left(A\right)}\right)=
\er^{G \times P_2}
\left(u^{\left(B\right)}, v^{\left(B\right)}\right).
\]

Also, note that because $k$ is even, each $k$ step walk in $G$ also
corresponds to a walk in $G \times P_2$ that starts/ends on the same side,
but alternates sides at each step.
Using (\ref{eq:gotItRight}) and the symmetry between $u^{(A)}v^{(B)}$ and $u^{(B)}v^{(A)}$, it suffices to sample length $k$ walks with estimated effective resistances satisfying for every edge $(u,v)\in G$
\[
\rrApprox_{uv} \geq \rr^{G \times P_2} \left(u^{(A)}, v^{(B)} \right).
\]
The rest of the algorithm follows similarly as in the case of odd $k$.

The extra term $\Theta(k\log n)$ in the overhead's runtime accounts for performing
a random walk of length $k$, i.e. after preprocessing in $O(n)$ time an neighboring edge
can be sampled using binary search in $O(\log n)$ time.
\end{proof}

This reduces the task of efficiently sampling edges
in  $G^{k}$ to producing good upper bounds for the effective resistances of an explicitly specified graph, either $G$
or $G \times P_2$.
In the next section we discuss this routine, with focus
on density-independent routines.

\section{Faster Density Independent Sparsification of Graphs}
\label{sec:mPlusStuff}

Note that the monomial sparsification routine only requires
a good distribution that dominates the effective resistances.
These effective resistances can in turn be computed w.r.t.
an approximate graph in a more efficient manner.

Our approach for obtaining density-independent routines
follow the approach given in~\cite{KoutisLP15}.
Namely, we aggressively make the graph more tree-like,
and build sparsifiers backwards, each leveraging access
to a sparsifier of a graph that is within a constant
factor of it.

The main algorithmic difficulty of designing density-independent schemes
is that numerically oriented approaches for estimating effective resistances
require $O(m \log{n})$ time.

Instead, a more useful method is to utilize low stretch spanning trees
to provide an upper bound on all leverage scores in terms of a tree's stretch.
The advantage of this approach is that the stretch of all edges in $G$ w.r.t.
a tree can be computed using lowest common ancestor queries in only $O(m)$
time~\cite{HarelT84}.
For a particular tree $T$, we define the stretch of an
edge $e$ w.r.t. $T$ as the sum of resistances over the
unique path $\mathcal{P}_T(e)$ in $T$ connecting $e$'s endpoints:
\[
str_{T}(e) \defeq \ww_e \sum_{e' \in \mathcal{P}_T(e)}
\frac{1}{\ww_{e'}}
\]
Extending this definition, the stretch of a graph $G(V,E,\ww)$ w.r.t. $T$ is given by
\[
str_{T}(G) \defeq \sum_{e \in E}
str_{T}(e)
\]

Our analysis relies on the following results:
\begin{lemma}\label{lem:TreeFacts}
\begin{enumerate}
\item \label{part:Transfer}
(Lemma 6.4. in~\cite{KoutisLP15})
If $G$ and $H$ are two graphs such that
$\LL_G \preceq \LL_H$, then for any
tree $T$, it holds $str_{T}(G)\leq str_{T}(H)$.

\item \label{part:EdgeCount}
If we have a tree $T \preceq G$, then we can construct
an $\eps$-sparsifier of $G$ with $O(\eps^{-2}
str_{T}(G) \log{n})$ edges in $O(m)$ time.
This is a direct consequence of Lemma~\ref{lem:sparsify},
as the stretch upper bounds leverage scores.

\item \label{part:LSST}
(Theorem 1 in~\cite{AbrahamN12})
For any graph $G$, we can obtain a tree with
total stretch $\Ohat(m \log{n})$ in
$\Ohat(m \log{n})$ time in the pointer machine model.

\item \label{part:LSSG}
(Lemma 5.9 in~\cite{CohenMPPX14:arxiv})
For any graph $G$ and any parameter $k$,
we can find in $\Ohat(m)$ time under the RAM model
a tree $T$ and a graph $\widehat{G}$ obtained by removing
$O(m / k)$ edges such that
$str_{T}(\widehat{G}) \leq \Ohat(m \log{n})$.
\end{enumerate}
\end{lemma}

We will use these tools to generate a sequence of graphs
based on a single low-stretch subgraph.
If we set $k \leftarrow \log^{\Theta(1)}n$ in Part~\ref{part:LSSG},
the number of edges omitted in $\widehat{G}$ can be sparsified
in $O(m)$ time using any of the sparsification
methods~\cite{KyngLPSS16,KyngPPS17}.
This leads to a scheme that start with
$\widehat{G}^{(0)} \defeq \widehat{G}$, and creates
a chain of graphs where the scaling factor of the tree increases:
\begin{equation}
\widehat{G}^{(i)} \defeq
\widehat{G} + 2^{i} \cdot T,
\label{eqn:GraphSequence}
\end{equation}
This sequence quickly leads to a graph whose stretch
is small enough that we can obtain an $O(1)$-sparsifier
in $O(m)$ time.
\begin{lemma}
\label{lem:LastStep}
For any $i \geq \Omega(\log\log{n})$,
an $O(1)$-sparsifier of $\widehat{G}^{(i)}$ with
$O(n \log{n})$ edges can be found in $O(m)$ time.
\end{lemma}

\begin{proof}
Multiplying $T$ by a factor of $2^{i}$
reduces the total stretch by the same factor.
From Lemma~\ref{lem:TreeFacts} Part~\ref{part:LSSG},
we get that the total stretch is bounded by
$\Ohat(2^{-i} m \log{n})$.
This means by Lemma~\ref{lem:TreeFacts}
Part~\ref{part:EdgeCount}, $\widehat{G}^{(i)}$
has an $O(1)$ sparsifier with $O(m \log^{-\Theta(1)}n)$
edges.
Invoking a nearly-linear time graph sparsification
algorithm on this graph then gives the result.
\end{proof}

We will leverage sparsifiers of $\widehat{G}^{(i + 1)}$
to construct iteratively, sparsifiers of $\widehat{G}^{(i)}$
using the subroutine shown in Algorithm~\ref{alg:TreeSparsify}.
We also state its guarantees formally below.

\begin{algorithm}[h]
	\caption{$\mathrm{TreeSparsify}(G, G^\prime, \kappa,
\eps)$}
	\label{alg:TreeSparsify}
	
	\textbf{Input:} Graph $G =(V, E, \ww)$ with
	$\kappa$-sparsifier $G'$, and error $\eps > 0$.
	
	\textbf{Output: } $\widetilde{G}$ that is an
	$\eps$-sparsifier of $G$.

	\begin{enumerate}
		\item Construct a low stretch spanning tree $T$
		of $G^\prime$.
        \item Compute an upper bound on all leverage scores $\ttautilde$ of $G$ using~\cite{HarelT84}
		\item Sample $O(\eps^{-2} \log{n} \cdot str_T(G))$
		edges of $G$ by $\mathrm{IdealSample}(G, \eps, \ttautilde)$ (c.f. Algorithm~\ref{alg:IdealSample}).
	\end{enumerate}
\end{algorithm}

\begin{lemma}
\label{lem:TreeSparsify}
Given a $\kappa$-sparsifier $G'$ of $G$ and $\eps > 0$,
$\mathrm{TreeSparsify}(G, G', \kappa,
\eps)$ produces an $\eps$-sparsifier of $G$ with at most $\Otil(\eps^{-2} \log^2{n} \cdot \kappa \abs{E(G')})$ edges in $\Otil(m + \eps^{-2} \log^3{n} \cdot \kappa \abs{E(G')})$ time.
\end{lemma}

\begin{proof}
To apply Lemma~\ref{lem:sparsify}, we have to compute a vector $\rrApprox\geq\er^{G}$ and give an upper bound on $\langle \ww, \rrApprox \rangle$. Since $\LL_T\preceq\LL_{G^\prime}\preceq\kappa\LL_{G}$ it follows that
\begin{equation}\label{eq:rApproxUB}
\rrApprox\defeq\kappa\cdot\er^{T}\geq\er^{G}.
\end{equation}
Moreover, by combining (\ref{eq:rApproxUB}), $\LL_G\preceq\LL_{G^\prime}$, Lemma~\ref{lem:TreeFacts}, Part~\ref{part:Transfer} and Part~\ref{part:LSST} we obtain
\[
\langle\ww,\rrApprox\rangle=\kappa\cdot str_{T}(G)\leq\kappa\cdot str_{T}(G^{\prime})=\Ohat(\kappa\abs{E(G')}\log{n}).
\]
This yields the overall edge count and runtime.
\end{proof}

We present two density-independent sparsification algorithms
that iteratively construct sparsifiers backwards from
$\widehat{G}^{(k)}$ by:
\begin{enumerate}
\item Creating a crude $\epsilon / 2$-sparsifier of
$\widehat{G}^{(i)}$, $G'^{(i)}$ with $\Ohat(\epsilon^{-2}
n \log^3{n})$ edges
using \textsc{TreeSparsify} with $\widetilde{G}^{(i + 1)}$,
the sparsifier of $\widehat{G}^{(i + 1)}$ constructed
in the previous step as guide.
\item Further sparsify this crude sparsifier, $G'^{(i)}$, with
error $\epsilon / 2$ to from $\widetilde{G}^{(i)}$, which
has the desired edge count, and will be used in the next step.
\end{enumerate}

We refrain from providing pseudocode of these steps because
of the subtle differences in the resulting algorithms.

The current fastest sparsification routines compute effective resistances via the Johnson-Lindenstrauss
transform~\cite{SpielmanS08:journal}, which in turn requires
the use of fast linear system solvers~\cite{KyngLPSS16}.
\begin{lemma}
\label{lem:ToolNumerical}
Given a graph $G$, we can compute $2$-approximations to
its effective resistances in $\Ohat(m \log{n} + n \log^{2} n)$ time.
\end{lemma}
This runtime bound can be obtained by letting the depth approach
$n$ in the proof of Theorem 1.2 on page 49 of~\cite{KyngLPSS16}.
The effective resistances can in turn be extracted from the
call to \textsc{Sparsify} made at $i = 0$ in the pseudocode
in Figure 11 on page 46.
We omit details on these steps in the hope that significantly
simpler sparsification routines with similar performances will
be developed.

Combining this with the sequence of graphs defined
in (\ref{eqn:GraphSequence}) gives:
\begin{corollary}\label{cor:SparsifyNumerical}
There is a routine that takes a weighted undirected graph
$G$ with $n$ vertices, $m$ edges,
an error $\epsilon > 0$,
and produces in $\Ohat(m + \eps^{-2} \log^4{n})$ time
an $\epsilon$-sparsifier of $G$ with
$O(\epsilon^{-2} n \log{n})$ edges, as well as leverage
score upper bounds that sum up to $\Ohat(n \log^{2}n)$.
\end{corollary}

\begin{proof}
Consider the sequence of matrices as defined in
(\ref{eqn:GraphSequence}).
The sparsifier for $\widehat{G}^{O(\log\log{n})}$
is given by Lemma~\ref{lem:LastStep}.

Then we can iteratively build $O(n\log{n})$ sized
$O(1)$-sparsifiers for $\widehat{G}^{(i)}$ all the
way up to $i = 1$.
The cost of invoking
$\mathrm{TreeSparisfy}(\widehat{G}^{(i)},
\widetilde{G}^{(i + 1)}, O(1), 2)$ at each step
is $\Ohat(m+n \log^{4}n)$, while the resulting sparsifier
has $\Ohat(n \log^{3}n)$ edges.
Lemma~\ref{lem:ToolNumerical} then turns this into
an $O(1)$-sparsifier for $\widehat{G}^{(i)}$.

At the last step of $i = 0$,
we invoke these same routines, but now with error $\epsilon$
to obtain the $\epsilon$-sparsifier.
Note that the effective resistance upper bounds are computed
during the call to $\mathrm{TreeSparisfy}$.
\end{proof}

The guarantees of this routine fits into the
requirements of the random walk sampling algorithm
from Section~\ref{sec:WalkSparsify} and yields the faster,
density-independent algorithm for sparsifying $G^{k}$,
that is our main result.

\begin{proof}(Of Theorem~\ref{thm:main})
The leverage scores upper bound obtained by Corollary~\ref{cor:SparsifyNumerical},
when combined with Theorem~\ref{thm:SparsifyGk}
produces an $\epsilon$-sparsifier for $G^{k}$ with
$\Ohat(\epsilon^{-2} k n \log^{3} n)$ edges in
$\Ohat(m + \eps^{-2} k^2 \log^4{n})$ time.
Sparsifying this graph once again using
Lemma~\ref{lem:ToolNumerical} then leads to the
main result as stated in Theorem~\ref{thm:main}.
\end{proof}

There are also purely combinatorial constructions of graph
sparsifiers based on spanners~\cite{KapralovP12,Koutis14,KyngPPS17}:
\begin{lemma}\label{lem:toolCombinatorial}
(Theorem 4.1. in~\cite{KyngPPS17})
Given a graph $G$ and an error $\eps > 0$,
we can compute an $\eps$-spectral sparsifier
of $G$ with $\Ohat(n \log^2{n})$ edges in $\Ohat(m \log^2{n})$ time.
\end{lemma}

The algorithm in Lemma~\ref{lem:toolCombinatorial}, applied to our sparsification scheme gives
\begin{corollary}
There is a combinatorial algorithm that for any graph
$G$ on $n$ vertices and $m$ edges, and any error $\epsilon > 0$,
produces in $\Ohat(m + n \log^{6}n)$ time an $\epsilon$-sparsifier with
$\Ohat(\eps^{-2} n \log^2{n})$ edges, as well
as leverage score upper bounds that sum up to
$\Ohat(n \log^3{n})$.
\end{corollary}

\begin{proof}
This is similar to the routine calling numerical sparsifiers
outlined in Corollary~\ref{cor:SparsifyNumerical}.
However, each of the $\widetilde{G}^{(i + 1)}$ now has $\Ohat(n \log^2{n})$ edges.
Hence, the first crude approximation $\widetilde{G'}^{(i)}$
has $\Ohat(n \log^4{n})$ edges.
Sparsifying it down to $\Ohat(n \log^2{n})$ edges
takes $\Ohat(n \log^6{n})$ time.
\end{proof}

\bibliographystyle{alpha}
\bibliography{ref}

\begin{appendix}

\section{Deferred Proofs}
\label{sec:deferred}	

We provide now some additional details on Lemmas
from Section~\ref{sec:WalkSparsify} that are direct
consequences of steps in previous works.
The total summation of the sampling weights follows from a
summation identical to the special case of uniform
sampling, as presented in~\cite[Lemma 29]{ChengCLPT15}. More precisely, by evaluating the total weights of all random walks that involve a particular edge $e\in G$.

\TotalWeights*

\begin{proof}
	We first show by induction that the total weights of
	all length $k$ walks whose $i$\textsuperscript{th} edge is $e$
	is exactly $\ww_e$.

	The base case of $k = 1$ is trivial as only $e$ is
	a length $1$ walk between $u_0$ and $u_1$.
	
	The inductive case of $k > 1$ has two cases: $i > 0$ or $i < k  - 1$.
	We consider the $i > 0$ case only, as the other one follows by symmetry.
	Expanding the weight of a length $k$ walk gives:
	\[
	\ww\left( u_0, u_1, \ldots u_k \right)
	= \ww\left( u_0, u_1, \ldots u_{k - 1} \right) \frac{\AA_{u_{k - 1} u_{k}}}{\dd_{u_{k - 1}}}.
	\]
	The fact that $i < k - 1$ means that $u_{k}$ can be any neighbor
	of $u_{k - 1}$, leading to a sum that cancels the $\dd_{u_{k - 1}}$
	term in the denominator.
	Formally:
	\begin{align*}
	\sum_{\begin{subarray}{c}(u_0, u_1, \ldots u_k)\\e = u_{i} u_{i + 1}\end{subarray}}
	\ww\left( u_0, u_1, \ldots u_k \right)
	& = \sum_{\begin{subarray}{c}(u_0, u_1, \ldots u_k)\\e = u_{i} u_{i + 1}\end{subarray}}
	\ww\left( u_0, u_1, \ldots u_{k - 1} \right)
	\sum_{u_k}\frac{\AA_{u_{k - 1} u_{k}}}{\dd_{u_{k - 1}}}\\
	& = \sum_{\begin{subarray}{c}(u_0, u_1, \ldots u_{k - 1})\\e = u_{i} u_{i + 1}\end{subarray}}
	\ww\left( u_0, u_1, \ldots u_{k - 1} \right).
	\end{align*}
	The result then follows from the inductive hypothesis applied
	to walks of length $k - 1$ that have edge $i$ as $e$.

The proof then uses a double counting argument that breaks the summation over the edge $(u_{i}, u_{i + 1})$,
and by noting that the choice over $i$ implies that each
edge is picked exactly $k$ times.
We can rewrite the original summation as:
\[
\sum_{e} \sum_{0 \leq i < k} \rrApprox_{e}
\cdot \left(
\sum_{\begin{subarray}{c}\left(u_0, u_1, \ldots, u_k \right)\\
u_{i}u_{i + 1} = e\end{subarray}}
	\ww_{u_0, u_1, \ldots, u_k}\right).
\]
The proof above gives that the term within the bracket
is $\ww_e$.
So the summation over $i$ is just an extra factor of $k$,
by which we obtain the result.

\end{proof}	

The equivalence of effective resistances in $G^{2}$
and $G \times P_2$ requires the definition of Schur
complements.

\begin{definition}[Schur Complement]
The Schur Complement of a symmetric matrix in block form:
$\MM = \begin{pmatrix}
\MM_{[F,F]} & \MM_{[F,C]} \\ \MM_{[C,F]} & \MM_{[C,C]}
\end{pmatrix}$
that removes the block $F$ is:
\[
Sc\left(\MM, F\right)
\defeq \MM_{\left[C,C\right]} - \MM_{\left[C,F\right]} \MM_{\left[F,F\right]}^{-1} \MM_{\left[F,C\right]}.\]
\end{definition}

It can be seen that, if $\MM$ was the Laplacian of a graph $G$,
$Sc(M,F)$ is the Laplacian of the graph $G^{S}$,
which can be formed using the following iterative process:
\begin{itemize}
\item Iteratively for all vertices $u \in F$
(The set of vertices represented by the columns in $F$ of $\MM$.)
\begin{itemize}
\item For all pairs of edges $uv_1$ and $uv_2$ in the current graph
(including edges added in previous steps),
delete them and add the edge $v_1v_2$ with weight
\[
\ww_{uv_1} \ww_{uv_2} / \dd_u,
\]
where $\dd_u$ is the weighted degree of $u$
(once again w.r.t. the current graph).
\item Delete $v$
\end{itemize}

\end{itemize}

\begin{lemma} \label{lem:schur}
\label{lem_propER} For every vector $\zz=\left(\begin{array}{c}
\zz_{1}\\
0
\end{array}\right)$ it holds that
\[
\zz_{1}^{T}\left(\DD-\AA\DD^{-1}\AA\right)^{\dag}z_{1}=\left(\begin{array}{cc}
\zz_{1}^{T} & 0^{T}\end{array}\right)\left(\begin{array}{cc}
\DD & -\AA\\
-\AA & \DD
\end{array}\right)^{\dag}\left(\begin{array}{c}
\zz_{1}\\
0
\end{array}\right).
\]
By symmetry for any vector $\zz=\left(\begin{array}{c}
0\\
\zz_{2}
\end{array}\right)$ it holds that
\[
\zz_{2}^{T}\left(\DD-\AA\DD^{-1}\AA\right)^{\dag}\zz_{2}
=\left(\begin{array}{cc}
0^{T} & \zz_{2}^{T}
\end{array}\right)
\left(\begin{array}{cc}
\DD & -\AA\\
-\AA & \DD
\end{array}\right)^{\dag}
\left(\begin{array}{c}
0\\
\zz_{2}
\end{array}\right).
\]
In particular, the effective resistances are maintained under Schur
complement.
\end{lemma}

\begin{proof}
Consider the linear system
\[
\left(\begin{array}{cc}
\DD & -\AA\\
-\AA & \DD
\end{array}\right)\left(\begin{array}{c}
\xx\\
\yy
\end{array}\right)=\left(\begin{array}{c}
\zz_{1}\\
\zz_{2}
\end{array}\right)\iff\begin{array}{c}
\DD \xx-\AA y=\zz_{1}\\
-\AA \xx+\DD y=\zz_{2}
\end{array}\iff\begin{array}{c}
\xx=\DD^{-1}\left(\zz_{1}+\AA \yy\right)\\
\yy=\DD^{-1}\left(\zz_{2}+\AA \xx\right)
\end{array}.
\]
Since $\zz_{2}=0$, we have
\[
\begin{array}{c}
\DD \xx=\zz_{1}+\AA\DD^{-1}\AA \xx\\
\yy=\DD^{-1}\AA \xx
\end{array}
\Longrightarrow
\xx=\left(\DD-\AA\DD^{-1}\AA\right)^{\dag}\zz_{1}.
\]
and thus
\[
\left(\begin{array}{cc}
\zz_{1}^{T} & \zz_{2}^{T}\end{array}\right)\left(\begin{array}{c}
\xx\\
\yy
\end{array}\right)
=
\zz_{1}^{T} \xx
=\zz_{1}^{T}\left(\DD-\AA\DD^{-1}\AA\right)^{\dag}\zz_{1}.
\]
\end{proof}

\EREven*
\begin{proof}
Notice that
\[
\LL_{G^2} = \DD - \AA \DD^{-1} \AA
\]
is the Schur Complement of
\[
\LL_{G \times P_2} =
\begin{pmatrix} \DD & -\AA \\ -\AA & \DD \end{pmatrix}
\]
with respect to one half of the vertices, e.g. $V^{(B)}$.
So, the lemma follows from Lemma~{\ref{lem:schur}}.
\end{proof}

\end{appendix}

\end{document}